\documentclass[letterpaper,10pt,conference]{ieeeconf}
\IEEEoverridecommandlockouts
\usepackage{graphicx}
\usepackage{subcaption}
\usepackage{etoolbox} 
\usepackage{scalerel} 
\usepackage{siunitx}
\usepackage{pifont}
\usepackage{graphicx}
\usepackage{amsmath,amssymb,textcomp}
\usepackage{algorithm,algpseudocode} 
\usepackage{cite}
\usepackage{booktabs}
\usepackage{url}
\usepackage{paralist}
\usepackage{color}
\usepackage[T1]{fontenc}
\usepackage{stmaryrd}
\usepackage{epsfig} 
\usepackage{bm} 
\usepackage{listings,lstautogobble}
\usepackage{mathtools}
\usepackage{mathrsfs}
\usepackage{xspace}
\usepackage{verbatim} 
\usepackage{epstopdf}
\usepackage{listings}
\usepackage{parcolumns}
\usepackage{color}
\usepackage{mathtools}
\usepackage{mathrsfs}
\usepackage{xspace}
\usepackage{verbatim}
\usepackage[utf8]{inputenc}
\usepackage{calrsfs}
\usepackage{rotating,multirow}
\usepackage{caption}
\usepackage{mathtools}

\usepackage[most]{tcolorbox}

\usepackage[T1]{fontenc}
\usepackage[utf8]{inputenc}

\usepackage{amsthm}

\newtheorem{lemma}{Lemma}

\newtheoremstyle{mainstyle}
  {\topsep} 
  {\topsep} 
  {} 
  {} 
  {\bfseries} 
  {.} 
  {.25em} 
  {} 

\theoremstyle{mainstyle}

\newtheorem{definition}{Definition}
\newtheorem{proposition}{Proposition}

\allowdisplaybreaks

\newtheorem{myassumption}{Assumption} 

\renewcommand{\footnotesize}{\fontsize{8pt}{11pt}\selectfont}

\DeclareMathAlphabet{\mathcal}{OMS}{cmsy}{m}{n}

\DeclarePairedDelimiter{\abs}{\lvert}{\rvert}

\newcommand{\diag}{\text{diag}}

\newcommand{\zono}[1]{\langle #1 \rangle}

\def\eqref#1{(\ref{#1})}
\def\eqnref#1{(\ref{#1})}

\usepackage{titlesec}
\titlespacing{\section}{0pt}{0.5ex}{0.5ex}
\titlespacing{\subsection}{0pt}{0.3ex}{0.3ex}
\titlespacing{\subsubsection}{0pt}{0.1ex}{0.1ex}

\def\colsep{\arraycolsep}
\def\t{\mathtt{T}}
\def\rank{\text{rank}}
\def\ol{\overline}

\let\OLDthebibliography\thebibliography
\renewcommand\thebibliography[1]{
  \OLDthebibliography{#1}
  \setlength{\parskip}{-1pt}
  \setlength{\itemsep}{-1pt plus 0.3ex}
}

\title{\LARGE \bf Data-driven Set-based Estimation of Polynomial Systems \\ with Application to SIR Epidemics
}

\author{Amr Alanwar$^{*,1}$, Muhammad Umar B. Niazi$^{*,2}$, and Karl H. Johansson$^{2}$
\thanks{$^*$Authors with equal contributions and ordered alphabetically.}
\thanks{$^1$Department of Computer Science \& Electrical Engineering, Jacobs University, Bremen, Germany (Email: {a.alanwar@jacobs-university.de}). $^2$Division of Decision and Control Systems and Digital Futures, EECS, KTH Royal Institute of Technology, Sweden (Emails: \{{mubniazi, kallej}\}{@kth.se}).}}
\begin{document}

\maketitle

\begin{abstract}
This paper proposes a data-driven set-based estimation algorithm for a class of nonlinear systems with polynomial nonlinearities. Using the system's input-output data, the proposed method computes a set that guarantees the inclusion of the system's state in real-time. Although the system is assumed to be a polynomial type, the exact polynomial functions, and their coefficients are assumed to be unknown. To this end, the estimator relies on offline and online phases. The offline phase utilizes past input-output data to estimate a set of possible coefficients of the polynomial system. Then, using this estimated set of coefficients and the side information about the system, the online phase provides a set estimate of the state. Finally, the proposed methodology is evaluated through its application on SIR (Susceptible, Infected, Recovered) epidemic model.
\end{abstract}
\section{Introduction}

For monitoring and control of dynamical systems, the knowledge of the state is undeniably crucial. Typically, observer design techniques rely on the system's model for estimating the state in real-time. However, as modern engineering systems are becoming increasingly complex, developing accurate models to describe a system's behavior is quite challenging \cite{conf:deepctherom, huang2021robust}. This motivates the development of a data-driven approach for set-based state estimation.

Given bounded uncertainties and measurement noise, estimating a set that guarantees the inclusion of the true system state at each time step is a classical problem studied in \cite{conf:setmem1971}. 
Several set-based estimation approaches have been presented in the literature. The interval observers \cite{efimov2013, conf:combastelinterval2, conf:combastelinterval1} utilize two observers to estimate the upper and lower bounds of the state trajectory and rely on the monotonicity properties of the estimation error dynamics to ensure that the true state remains inside the estimated bounds. Secondly, the set-membership observers \cite{conf:orthotope,conf:ierardi_setmem,conf:ourdiffusion} intersect the state-space regions consistent with the model with those obtained from the online measurements to estimate the current state set. Finally, zonotopic filtering \cite{alamo2005, le2013, rego2020, rego2021} provides set-based state estimation using interval arithmetic and zonotopes under bounded uncertainties in the model.
See \cite{conf:AlthoffJagatjournal, paula2022} for a comprehensive review of the existing set-based estimators.

The limitation of the existing methods for set-based estimation is their assumption of an a priori known model. Such a model is not available in many applications or too costly to develop and identify. Furthermore, relying on inaccurate models may violate the formal guarantees of the system.
The data-driven paradigm is therefore gaining precedence over the model-based paradigm because of the possibility to obtain huge amounts of data from the system thanks to the advancements in sensor technologies. Recently, several studies are dedicated to data-driven reachability analysis~\cite{conf:LearningDensity, conf:onthefly,  conf:murat_christoffel, conf:activelearning1, conf:uncertain, conf:stochasticreach, conf:koopmanblack, conf:l4dc,conf:zpc} that overcome the limitation of prior model knowledge. However, to the best of our knowledge, only one work presented a set-based estimation technique for linear systems \cite{conf:data_driven_set_based} given the offline data and online measurements. 

In this paper, we propose a set-based estimator for nonlinear polynomial systems by extending \cite{conf:data_driven_set_based}. We provide formal guarantees for the proposed data-driven set-based estimation method and show its effectiveness on the application of compartmental SIR (Susceptible, Infected, Recovered) epidemic process. The set-based estimation for epidemics is essential as it provides formal guarantees on the bounds of the infected population in the presence of uncertainties and discrepancies in the data. The existing set-based estimators \cite{bliman2016, degue2020, efimov2021} for compartmental epidemics rely on a model and a priori knowledge of some of its parameters. In contrast, our method does not assume such an a priori knowledge. However, if the system's model structure is known, our method can incorporate it as a side information to further refine the set-based estimation results.


\section{Preliminaries and Problem Statement} 
\label{sec:preliminaries}


\subsection{Notations and Set Representations}

The set of real and natural numbers are denoted as $\mathbb{R}$ and $\mathbb{N}$, respectively, and $\mathbb{N}_0 = \mathbb{N}\cup \{0\}$. The transpose and Moore-Penrose pseudoinverse of a matrix $X$ are denoted as $X^\t$ and $X^\dagger$, respectively.
If $X\in\mathbb{R}^{n\times m}$ is full-row rank, then $X^\dagger = X^\t (XX^\t)^{-1}$ is the right inverse, i.e., $XX^\dagger = I_n$, where $I_n\in\mathbb{R}^{n\times n}$ is the identity matrix. 
The standard unit vector $e_i$ is the $i$-th column of $I_n$, for an appropriate $n$, and $E_i = \diag(e_i)$. The $i$-th element of a vector or list $A$ is denoted by $A^{(i)}$. 

\begin{definition}(Zonotope \cite{conf:zono1998})
\label{def:zonotope} 
Given a center $c \in \mathbb{R}^n$ and $\xi \in \mathbb{N}$ generator vectors $g^{(1)}, \dots, g^{(\xi)}\in\mathbb{R}^n$, a zonotope is the set
    \begin{equation*}
      \mathcal{Z} = \Big\{ x \in \mathbb{R}^n \,\big|\, x = c + \sum_{i=1}^\xi \beta_i \, g^{(i)} \, ,
      -1 \leq \beta_i \leq 1 \Big\}
    \end{equation*} 
denoted as $\mathcal{Z} = \zono{c,G}$ with $G=[\colsep2pt\begin{array}{ccc} g^{(1)} & \dots & g^{(\xi)}\end{array}]$. 
\end{definition}

A linear map $L:\mathbb{R}^{n}\rightarrow\mathbb{R}^m$ applied to a zonotope $\mathcal{Z}$ yields $L\mathcal{Z}=\zono{Lc,LG}$. Given two zonotopes $\mathcal{Z}_1=\zono{c_1,G_1}$ and $\mathcal{Z}_2=\zono{c_2,G_2}$ that are subsets of $\mathbb{R}^n$, their Minkowski sum is given by
$
\mathcal{Z}_1+\mathcal{Z}_2 = \zono{c_1+c_2,[\colsep=2pt\begin{array}{cc} G_1 & G_2 \end{array}]}.
$
Note that $\mathcal{Z}_1 - \mathcal{Z}_2$ means $\mathcal{Z}_1 + (-1 \mathcal{Z}_2)$. 
  
\begin{definition}(Constrained zonotope \cite{conf:const_zono}) \label{df:contzono}
Given a center $\ol{c} \in \mathbb{R}^n$ and $\xi \in \mathbb{N}$ generator vectors $\ol{g}^{(1)}, \dots, \ol{g}^{(\xi)}\in\mathbb{R}^n$, and constraints $A \in $ $\mathbb{R}^{n_c \times 
\xi}$ and $b \in \mathbb{R}^{n_c}$, a constrained zonotope is the set
\begin{equation}\label{eq:conszono}
    \mathcal{C} {=} \Big\{ x\in\mathbb{R}^n \,\big|\, x=\ol{c}+\sum_{i=1}^\xi \beta_i \ \ol{g}^{(i)}, \ A \beta=b,\nonumber \\ 
    -1 \leq \beta_i \leq 1 \Big\}
\end{equation}
denoted as $\mathcal{C} = \zono{\ol{c},\ol{G},A,b}$ with $\ol{G}=[\colsep2pt\begin{array}{ccc} \ol{g}^{(1)} & \dots & \ol{g}^{(\xi)}\end{array}]$. 
\end{definition}

\begin{definition}(Matrix zonotope \cite[p.52]{conf:thesisalthoff}) \label{def:mat_zonotope}
Given a center $C \in \mathbb{R}^{n\times k}$ and $\xi \in \mathbb{N}$ generator matrices ${G}^{(1)},\dots,G^{(\xi)} \in \mathbb{R}^{n \times k}$, a matrix zonotope is the set
  \begin{equation*}
    \mathcal{M} = 
    \Big\{   X \in \mathbb{R}^{n\times k} 
             \,\big|\,
             X = C + \sum_{i=1}^\xi {\beta}_i \, {G}^{(i)} \, 
             ,
             -1 \leq {\beta}_i \leq 1 
    \Big\}
  \end{equation*}
denoted as $\mathcal{M} = \zono{C,\mathcal{G}}$ with $\mathcal{G}=\{ {G}^{(1)}, \dots, {G}^{(\xi)} \}$.
\end{definition}

A linear map $L:\mathbb{R}^n \rightarrow \mathbb{R}^m$ applied to a matrix zonotope $\mathcal{M}=\zono{C,\mathcal{G}}$ yields $L\mathcal{M} = \zono{LC,L\mathcal{G}}$. Given two matrix zonotopes $\mathcal{M}_1=\zono{C_1,\mathcal{G}_1}$ and $\mathcal{M}_2=\zono{C_2,\mathcal{G}_2}$ that are subsets of $\mathbb{R}^{n\times k}$, their Minkowski sum is given by
$
\mathcal{M}_1 + \mathcal{M}_2 = \zono{C_1+C_2,\{\mathcal{G}_1,\mathcal{G}_2\}}.
$

\begin{definition}(Constrained matrix zonotope \cite{conf:ourjournal}) \label{def:conmatzonotopes}  
Given a center $\overline{C} \in \mathbb{R}^{n\times k}$, generator matrices $\overline{G}^{(1)}, \dots, \overline{G}^{(\xi)}\in \mathbb{R}^{n \times k }$, and constraints $A^{(1)}, \dots , A^{(\xi)} \in \mathbb{R}^{n_c \times n_a}$ and $B \in \mathbb{R}^{n_c \times n_a}$, a constrained matrix zonotope is the set
\[\arraycolsep=1pt\begin{array}{cl}
	\mathcal{N} = \Big\{ X \in \mathbb{R}^{n\times p} \,\big|\,& X = \overline{C} + \sum_{i=1}^{\xi} \beta_i \, \overline{G}^{(i)}, ~\text{where} \\
	& \sum_{i=1}^{\xi} \beta_i A^{(i)} = B \,,  -1 \leq \beta_i \leq 1 \Big\}
\end{array}\]
denoted as $\mathcal{N} = \zono{\overline{C}, \overline{\mathcal{G}}, \mathcal{A}, B}$ with $\overline{\mathcal{G}}=\{\overline{G}^{(1)},\dots,\overline{G}^{(\xi)}\}$ and $\mathcal{A}=\{A^{(1)},\dots,A^{(\xi)}\}.$
\end{definition}

\begin{definition}(Interval matrix \cite[p. 42]{conf:thesisalthoff}) \label{def:intmat}
An interval matrix $\mathcal{I} = \left[\, \underline{I},\overline{I} \, \right]$ has intervals as its entries, where the left and right limits $\underline{I},\overline{I}\in \mathbb{R}^{n \times k}$ are such that $\underline{I}\leq \overline{I}$ element-wise.
\end{definition} 

To over-approximate a zonotope $\mathcal{Z}=\zono{c,[\colsep=1pt\begin{array}{ccc} g^{(1)} & \dots & g^{(\xi)}\end{array}]}$ by an interval $\mathcal{I} = [\,\underline{i},\bar{i}\,]$, we write $\mathcal{I}=\text{int}(\mathcal{Z})$ that is computed as
\begin{equation}\label{eq:zono2int}
    \bar{i} = c + \sum_{i=1}^{\xi} \abs{g^{(i)}}, \quad
    \underline{i} = c -  \sum_{i=1}^{\xi} \abs{g^{(i)}}. 
\end{equation}
We compute the inverse of an interval matrix by following \cite[Theorem 2.40]{conf:inverseInterval}, but other types of inverses provided in \cite{conf:inverseIntervalSurvey} can also be used. The pseudoinverse of an interval matrix $\mathcal{I}$ will also be denoted as $\mathcal{I}^\dagger$. We denote the interval vector (column) $i$ of an interval matrix $\mathcal{I}$ by $\mathcal{I}(:,i)$.

\subsection{Polynomial Systems}
Consider a discrete-time nonlinear system
\begin{subequations}\label{eq:sysnonlin_poly}
\begin{align}
    x(k+1) &= f_p(x(k),u(k))+ w(k) \label{eq:underlying_system} \\
    y(k) &= H x(k) + v(k) \label{eq:observations}
\end{align}
\end{subequations}
where $f_p:\mathbb{R}^{n_x}\times\mathbb{R}^{n_u} \rightarrow \mathbb{R}^{n_x}$ is a polynomial nonlinearity with $p\in\mathbb{R}^{n_p}$ the vector of parameters,
$w(k) \in \mathcal{Z}_w \subset \mathbb{R}^{n_x}$ is the process noise bounded by the zonotope $\mathcal{Z}_w$, ${u(k) \in \mathbb{R}^{n_u}}$ is the known input, $y(k)  \in \mathbb{R}^{n_y}$ is the output measured by sensors with $n_y\leq n_x$, and $v(k) \in \mathcal{Z}_{v} \subset \mathbb{R}^{n_y}$ is the measurement noise of sensors bounded by the zonotope $\mathcal{Z}_v$.
Without loss of generality, we assume that the output matrix $H\in\mathbb{R}^{n_y\times n_x}$ is full-row rank, i.e., $\rank(H)=n_y$. Moreover, the initial state $x(0)\in\mathcal{X}_0$, for some known $\mathcal{X}_0\subset\mathbb{R}^{n_x}$, and the system \eqref{eq:sysnonlin_poly} is assumed to satisfy the observability rank condition in the sense of \cite{albertini2002, hanba2009}.

In the interest of clarity, we will sometimes omit $k$ as the argument of signal variables, however, the dependence on $k$ should be understood implicitly. 
Let $n=n_x+n_u$ and
$
\zeta(k)=[\colsep=2pt\begin{array}{cc} x(k)^\t & u(k)^\t \end{array}]^\t=[\colsep=2pt\begin{array}{ccc} \zeta_1(k) & \dots & \zeta_n(k) \end{array}]^\t \in \mathbb{R}^n.
$
By a polynomial system, we mean that $f_p(\zeta) \in \mathbb{R}[\zeta]^{n_x}$ is a polynomial nonlinearity, where $\mathbb{R}[\zeta]^{n_x}$ is an $n_x$-dimensional vector with entries in $\mathbb{R}[\zeta]$, which is the set of all polynomials in the variables $\zeta_1,\dots,\zeta_n$ of some degree $d>0$ given by
\begin{align*}
     f_p^{(i)}(\zeta)  = \sum_{j=1}^m \theta_j \zeta_1^{\alpha_{j,1}} \zeta_2^{\alpha_{j,2}} \dots \zeta_n^{\alpha_{j,n}} 
\end{align*}
with $m$ the number of terms in $f_p^{(i)}(\zeta)$, $\theta_j \in \mathbb{R}$ the coefficients, and $\alpha_j = [\colsep=2pt\begin{array}{ccc} \alpha_{j,1} & \dots & \alpha_{j,n} \end{array}]^\t\in\mathbb{N}_0^{n}$ the vectors of exponents with $\sum_{i=1}^n \alpha_{j,i} \leq d$, for every $j\in\{1,\dots,m\}$. 


\subsection{Problem Statement}

Given the input vector $u(k)$, output vector $y(k)$, output matrix $H$, and noise zonotopes $\mathcal{Z}_w$ and $\mathcal{Z}_v$, our main goal is to obtain a set-based estimate $\hat{\mathcal{R}}_k$ that guarantees the inclusion of the true state, i.e., $x(k)\in\hat{\mathcal{R}}_k$ given that $n_x$ is known. Also, we aim to estimate the set of possible coefficients of the polynomial function $f_p(x(k),u(k))$. 

\section{Set-based Estimation Algorithm}\label{sec:alg}

We can write $f_p(\zeta)$ as follows (see \cite{conf:ourjournal} and \cite{conf:polyDissipat})
\begin{align}
    f_p(\zeta) &{=} \Theta_p \, h(\zeta) 
    \label{eq:pcg}
\end{align}
where 
$h(\zeta) \in \mathbb{R}[\zeta]^{m_a}$ contains at least all the monomials present in $f_p(\zeta)$. These monomials can be included in $h(\zeta)$ if, for instance, the upper bound on the degree of polynomials in $f_p(\zeta)$ is known. Moreover, if the polynomial function $f_p(\zeta)$ is known, then $h(\zeta)$ contains all the monomials of $f_p(\zeta)$. The matrix ${\Theta_p \in \mathbb{R}^{n_x \times m_a}}$ contains the unknown coefficients of the monomials in $h(\zeta)$.

The proposed set-based estimator consists of two phases: offline and online, which are detailed below.

\subsection{Offline Phase} \label{sec:offline}

In this subsection, we consider an offline phase, where it is assumed that an experiment is conducted and the data on the input $u(k)$ and the output $z(k)$ trajectories is collected for $k=0,1,\dots,T$, where 
\begin{align}
  z(k)  &= H x(k) + \gamma(k)
  \label{eq:training_measurements}
\end{align}
  with $\gamma(k)\in \mathcal{Z}_{\gamma} \subset \mathbb{R}^{n_y}$ representing the noise
  bounded by the zonotope $\mathcal{Z}_{\gamma} = \zono{c_{\gamma},G_{\gamma}}$.
Notice that $z(k)\in\mathbb{R}^{n_y}$ denotes the data collected offline and, to avoid confusion, is distinguished from $y(k)\in\mathbb{R}^{n_y}$, which denotes the online sensor measurements in the next subsection.
Given the above experiment, we obtain a sequence of noisy data and construct the following matrices of length $T$
\begin{equation} \label{eq:data_sequences}
\begin{array}{c}
  Z^{+} = \begin{bmatrix}z(1)&\dots&z(T)\end{bmatrix}, \;
  Z^{-} = \begin{bmatrix}z(0)&\dots&z(T-1)\end{bmatrix} \\ [0.5em]
  U^{-} = \begin{bmatrix}u(0)&\dots&u(T-1)\end{bmatrix}
\end{array}
\end{equation}
and let
$
 Z= \begin{bmatrix} Z^- & z(T)\end{bmatrix}
$.

\begin{myassumption} \label{assump:boundsx}
It is assumed that $\|x(k)\|_\infty\leq \kappa$, for every $k\in\{0,\dots,T\}$ and some known $\kappa>0$.
\end{myassumption}


  We aim to determine the mapping of the observation $Z^{+}$ and $Z^{-}$
  to the corresponding state-space region.
In other words, 
  we construct a zonotope $\mathcal{Z}_{x|z(k)} \subset \mathbb{R}^n$ 
  that contains all \textit{possible} $x(k) \in \mathbb{R}^{n_x}$
  given the measurement $z(k)\in\mathbb{R}^{n_y}$, output matrix $H\in\mathbb{R}^{n_y\times n_x}$ and 
  bounded noise $\gamma(k) \in \mathcal{Z}_{\gamma}$ 
  satisfying \eqnref{eq:training_measurements}.
Precisely, we construct the set
\begin{align*} 
    \mathcal{Z}_{x|z(k)} = \left\{ x(k) \in \mathbb{R}^n \,\big|\, 
        H x(k) = z(k) - \mathcal{Z}_{\gamma}
    \right\}
\end{align*}
from the following result inspired by \cite[Proposition 1]{conf:data_driven_set_based}. 
  \begin{lemma} \label{lem:measurement_zonotope}
    Let Assumption~\ref{assump:boundsx} hold.
    Then, given the measurement $z(k)\in\mathbb{R}^{n_y}$ satisfying \eqnref{eq:training_measurements} with bounded noise $\gamma(k) \in \mathcal{Z}_{\gamma} = \zono{c_{\gamma},G_{\gamma}}$, 
    the corresponding state vector $x(k)\in\mathbb{R}^{n_x}$
    is contained within the zonotope
    $ \mathcal{Z}_{x|z(k)} = \zono{c_{x|z(k)},G_{x|z(k)}},$
    where
    \begin{equation}
      \begin{split}
        c_{x|z(k)} &= H^\dagger \big( z(k) - c_{\gamma} \big) 
        \\
        G_{x|z(k)} &= \begin{bmatrix} 
          H^\dagger G_{\gamma} & \kappa (I_{n_x} - H^\dagger H) 
      \end{bmatrix}.
      \end{split}
      \label{eq:prop_1_eqn}
    \end{equation}
  \end{lemma}
  \begin{proof}
  By Assumption~\ref{assump:boundsx}, we have $-\kappa 1_{n_x} \leq x(k) \leq \kappa 1_{n_x}$ element-wise. Thus, the zonotope $\mathcal{Z}_\kappa=\zono{0_{n_x},\kappa I_{n_x}}$ contains $x(k)$. 
  Left-multiplying by $H^\dagger$, adding $x(k)$ on both sides, and rearranging \eqref{eq:training_measurements} gives
  \begin{equation} \label{eq:x-equation}
  x(k) = H^\dagger \left( z(k) - \gamma(k) \right) + (I_{n_x} - H^\dagger H ) x(k)
  \end{equation}
  which implies that
  $
  \mathcal{Z}_{x|z(k)} = H^\dagger \left( z(k) - \mathcal{Z}_\gamma \right) + (I_{n_x} - H^\dagger H ) \mathcal{Z}_\kappa.
  $
  Therefore, performing the linear transformations and Minkowski sum operations gives $\mathcal{Z}_{x|z(k)}$ with the center $c_{x|z(k)}$ and the generator $G_{x|z(k)}$ as in \eqref{eq:prop_1_eqn}.
  \end{proof}
  
Using the zonotope $\mathcal{Z}_{x|z(k)}$ for each sample $x(k)$, we obtain a matrix zonotope that provides the mapping of $Z^{+}$ and $Z^{-}$ to the state space. For $\gamma(k)$, consider the extension of the noise zonotopes $\mathcal{Z}_{\gamma}=\zono{c_\gamma,[g_\gamma^{(1)} \dots g_\gamma^{(\xi_\gamma)}]}$ to a matrix zonotope $\mathcal{M}_{\gamma}=\zono{C_\gamma,\mathcal{G}_\gamma}$, where
$C_\gamma = [\begin{array}{ccc} c_\gamma & \dots & c_\gamma \end{array}] \in \mathbb{R}^{n_y\times T}$ and
$\mathcal{G}_\gamma = \{G_\gamma^{(1)}, \dots, G_\gamma^{(\xi_\gamma T)}\}$
with
$G_\gamma^{(j+(i-1)T)} = [\colsep=2pt\begin{array}{ccc} 0_{n_y\times (j-1)} & g_\gamma^{(i)} & 0_{n_y \times (T-j)} \end{array}]$,
for $i=1,\dots,\xi_\gamma$ and $j=1,\dots,T$ and similarly we define $\mathcal{M}_{w}=\zono{C_w,\mathcal{G}_w}$. 

\begin{proposition} \label{prop:matrixzonoXZ}
Let Assumption~\ref{assump:boundsx} hold. Then, given the data $Z^+$ in \eqref{eq:data_sequences}, and noise matrix zonotope $\mathcal{M}_\gamma=\zono{C_\gamma,\mathcal{G}_\gamma}$, the unknown sequence of state $X^+ = [\colsep=2pt\begin{array}{ccc} x(1) & \dots & x(T) \end{array}]$ is contained within the zonotope $\mathcal{M}_{X|Z}^+=\zono{C_{X|Z}^+,\mathcal{G}_{X|Z}}$, where
\begin{equation}
    \colsep=1pt
    \begin{split}
        C_{X|Z}^+ &= H^\dagger (Z^+ - C_\gamma) \\
        \mathcal{G}_{X|Z} &= \{H^\dagger \mathcal{G}_\gamma, \kappa (I_{n_x}\!-H^\dagger H) 1_{n_x\times T} \}.
    \end{split}
    \label{eq:matrixzonoXZ}
\end{equation}
\end{proposition}
\begin{proof}
From \eqref{eq:x-equation}, we have
$
X^+ = H^\dagger (Z^+ - \Gamma^+) + (I-H^\dagger H) X^+
$
where $\Gamma^+ = [\colsep=2pt\begin{array}{ccc} \gamma(1) & \dots & \gamma(T) \end{array}]\in\mathcal{M}_{\gamma}$.
Due to Assumption~\ref{assump:boundsx}, we have $X^+\in\mathcal{M}_\kappa$ where $\mathcal{M}_\kappa$ is a matrix zonotope with center at $0$ and one generator $\kappa 1_{n_x\times T}$. Therefore, we over-approximate $X^+$ by
\[
\mathcal{M}_{X|Z}^+ = H^\dagger ( Z^+ - \mathcal{M}_\gamma ) + (I-H^\dagger H) \mathcal{M}_\kappa.
\] 
Hence, by applying the linear transformations and performing the Minkowski sum of the matrix zonotopes on the right-hand side of the above equation, we obtain \eqref{eq:matrixzonoXZ}.
\end{proof}

Let 
$
X^-=\left[\colsep=2pt\begin{array}{ccc} x(0) & \dots & x(T-1) \end{array}\right].
$ 
We over-approximate $X^-$ by $\mathcal{M}_{X|Z}^-$ by making use of Proposition \ref{prop:matrixzonoXZ} and $Z^-$. From \eqref{eq:pcg}, consider
\begin{align} \label{eq:bigomega}
\Omega(X^-,U^-)=\left[\colsep=2pt\begin{array}{ccc} h(x(0),u(0)) & \dots & h(x(T-1),u(T-1))\end{array}\right]
\end{align}
to be a matrix in $\mathbb{R}^{m_a\times T}$ with state and input trajectories substituted in $h(x(k),u(k))$, for $k=0,\dots,T-1$. By converting the matrix zonotope $\mathcal{M}_{X|Z}^-$ into interval matrix $\mathcal{I}_{X|Z}^- = \text{int}(\mathcal{M}_{X|Z}^-)$ and substitute in \eqref{eq:bigomega}, we obtain an interval matrix
\begin{equation*}
\begin{aligned}
\Omega(\mathcal{I}_{X|Z}^-,U^-) &{=} \left[\begin{matrix} h(\mathcal{I}_{X|Z}^-(:,0),u(0)) & \dots \end{matrix}\right.\\
& \qquad \left.\begin{matrix}
h(\mathcal{I}_{X|Z}^-(:,T-1),u(T-1))
\end{matrix}\right].
\end{aligned}
\end{equation*}


\begin{proposition}
\label{prop:sigmaM_p}
The matrix zonotope 
\begin{align}
    \mathcal{M}_{\Theta_p} = (\mathcal{M}^{+}_{X|Z} - \mathcal{M}_w) \Omega(\mathcal{I}_{X|Z}^-,U^-)^\dagger \label{eq:Msigma}
\end{align} 
  contains all matrices $\Theta_p$ that are consistent with the data $\{Z,U^-\}$ and the process noise matrix zonotope $\mathcal{M}_w$. 
\end{proposition}
\begin{proof}
From \eqref{eq:sysnonlin_poly} and \eqref{eq:pcg}, we have
    $X^+ = \Theta_p \Omega(X^-,U^-)  + W^- $
where $W^-=[\colsep=2pt\begin{array}{ccc} w(0) & \dots & w(T-1) \end{array}]$. The true $\Theta_p$ can be represented by a specific choice of $\beta^{(i)}$ in the matrix noise zonotope $\mathcal{M}_w$ which in turn results in the specific $W^- \in \mathcal{M}_w$. However, as we do not know the true $W^-$, we consider all matrices in the matrix zonotope $\mathcal{M}_w$ and compute the corresponding $\Theta_p$. Moreover, we do not have access to $X^+$, $X^-$, so we over-approximate $\Theta_p$ by considering the matrix zonotope $\mathcal{M}_{X|Z}^+$ in \eqref{eq:matrixzonoXZ} that bounds $X^+$ by Proposition~\ref{prop:matrixzonoXZ}. Finally, instead of $\Omega(X^-,U^-)$, we consider the interval matrix $\Omega(\mathcal{I}_{X|Z}^-,U^-)$ which gives \eqref{eq:Msigma}.
\end{proof}

\begin{algorithm}[t]
  \caption{Set-based estimation of polynomial systems}
  \label{alg:PolySetBasedEstimation}
  \textbf{Input}: Data matrices $\{Z,U^-\}$ of the polynomial system \eqref{eq:sysnonlin_poly}, initial set $\mathcal{X}_{0}$, process noise zonotope $\mathcal{Z}_w$ and matrix zonotope $\mathcal{M}_w$, online input $u(k)$ and measurement $y(k)$, for $k = 0, \dots,N-1$. \\
  \textbf{Output}: Estimated sets $\mathcal{N}_{\Theta_p}$ and $\hat{\mathcal{R}}_{k}$ 
 \begin{algorithmic}[1]
  \Statex \texttt{// Offline phase //}
 %
  \State Obtain $\mathcal{M}^{-}_{X|Z},\mathcal{M}^{+}_{X|Z}$ from Proposition \ref{prop:matrixzonoXZ} 
  \State $\mathcal{I}^-_{X|Z} = \text{int}(\mathcal{M}^{-}_{X|Z})$ 
    \State $\mathcal{M}_{\Theta_p} = (\mathcal{M}^{+}_{X|Z} - \mathcal{M}_w) \Omega(\mathcal{I}^{-}_{X|Z},U^-)^\dagger$\label{ln:algMsigma}
  \State $\overline{C}_{\Theta_p} = C_{\Theta_p}$ \label{ln:C}
  \State $\overline{G}_{\Theta_p}^{(i)} = G_{\Theta_p}^{(i)}$, for $i=1, \dots, \xi_{\Theta_p}$ \label{ln:G1}
  \State $\overline{G}_{\Theta_p}^{(i)} = 0$, for $i=\xi_{\Theta_p} +1, \dots,\xi_{\Theta_p} +n_xm_a$ \label{ln:G2}
  \State $\overline{\mathcal{G}}_{\Theta_p}=\{  \overline{G}_{\Theta_p}^{(1)},\dots,\overline{G}_{\Theta_p}^{(\xi_{\Theta_p} + n_xm_a)}\}$ \label{ln:Gset}
  \State  $A_{\Theta_p}^{(i)} = \bar{Q}G_{\Theta_p}^{(i)}$, for $i=1, \dots, \xi_{\Theta_p}$ \label{ln:A1}
    \State  $A_{\Theta_p}^{(\xi T+k)} = - E_i R E_j$, where $(i,j)\mapsto k=1,\dots,n_x m_a$, for $i=1, \dots, n_x$ and $j=1, \dots, m_a$ \label{ln:A2}
      \State $\mathcal{A}_{\Theta_p}=\{  A_{_{\Theta_p}}^{(1)}, \dots, A_{\Theta_p}^{(\xi_{\Theta_p}+ n_x m_a)}\}$, $B_{\Theta_p} =  \bar{Y} - \bar{Q} C_{\Theta_p}$ \label{ln:Aset}
    \State $\mathcal{N}_{\Theta_p} = \zono{\overline{C}_{\Theta_p},\overline{\mathcal{G}}_{\Theta_p},\mathcal{A}_{\Theta_p},B_{\Theta_p}}$  \label{ln:N}
    \Statex  \texttt{// Online phase //} 
  \State $\hat{\mathcal{R}}_{0} =\mathcal{X}_{0}$  \label{ln:R_0}
  \For{$k = 0:N-1$}
  \State $\hat{\mathcal{I}}_{k} = \text{int}(\hat{\mathcal{R}}_k)$ \label{ln:algI_i}
  \State $\tilde{\mathcal{R}}_{k+1} = \mathcal{N}_{\Theta_p}\Omega(\hat{\mathcal{I}}_{k},u(k)) +  \mathcal{Z}_w$ \label{ln:timeupdate}
  \State Obtain $\hat{\mathcal{R}}_{k+1}$ from Proposition~\ref{prop:con_zonotopes} 
  \label{ln:measupdate} 
  \EndFor
  \end{algorithmic}
\end{algorithm}

The above proposition proves that the matrix zonotope $\mathcal{M}_{\Theta_p}$ contains the true $\Theta_p$, however, it might be conservative. To obtain a tighter bound on the true $\Theta_p$, one might resort to prior knowledge about the system, such as bounds on the parameters and zero-pattern structure of $\Theta_p$, and incorporate it in the estimated set by using a constrained matrix zonotope. Such knowledge can be attained by studying the physics of the system or the environment in which the system operates. It would be beneficial to make use of this side information to have less conservative estimated set bounds. We consider the proposed approach in \cite{conf:ourjournal} to incorporate prior information about the unknown coefficients like decoupled dynamics, partial knowledge, or prior bounds on entries of the unknown coefficients matrices. We consider any side information that can be formulated as 
\begin{equation}
    | \bar{Q} \Theta_p - \bar{Y} | \leq \bar{R} \label{eq:sideinfo}
\end{equation}
where $\bar{Q} \in \mathbb{R}^{n_s\times n_x}$, $\bar{Y} \in \mathbb{R}^{n_s\times m_a}$, and $\bar{R} \in \mathbb{R}^{n_s\times m_a}$ are matrices defining the side information for the true $\Theta_p\in\mathbb{R}^{n_x\times m_a}$. Here, $|\cdot|$ and $\leq$ are element-wise operators.
    
After obtaining a matrix zonotope $\mathcal{M}_{\Theta_p}$ that bounds the set of unknown coefficients,  we utilize  the constrained matrix zonotopes to incorporate side information in the set-based estimation. Specifically, we compute the constrained matrix zonotope $\mathcal{N}_{\Theta_p}$ as in lines \ref{ln:C}-\ref{ln:N} of Algorithm \ref{alg:PolySetBasedEstimation}. We use the same center of the $\mathcal{M}_{\Theta_p}$ in line \ref{ln:C} and append zeros to its list of generators in lines \ref{ln:G1} and \ref{ln:G2}. Then, we compute the list of constrained matrices in lines \ref{ln:A1} to \ref{ln:Aset}. Our computations to tighten the set of the unknown coefficients are adaptations of the theory in \cite{conf:ourjournal} and can be easily proved.

\subsection{Online Phase} \label{sec:online}

We present the online estimation phase by considering the system \eqref{eq:sysnonlin_poly} with measurements $y(k)$.
This phase consists of a time update step computing $\tilde{\mathcal{R}}_{k}$ and a measurement update step computing $\hat{\mathcal{R}}_{k}$ as described in Algorithm~\ref{alg:PolySetBasedEstimation}.

\subsubsection{Time update}

We first initialize the measurement update set $\hat{\mathcal{R}}_0$ in line~\ref{ln:R_0}. Then, at each time step $k=0,\dots,N-1$, we convert the estimated set into an interval $\hat{\mathcal{I}}_k$ in line \ref{ln:algI_i} by using \eqref{eq:zono2int}. Given the current input $u(k)$ and the interval matrix $\hat{\mathcal{I}}_k$, we substitute in the list of monomials $\Omega(\hat{\mathcal{I}}_{k},u(k))$. Then, in line~\ref{ln:timeupdate}, we propagate ahead the measurement update set $\hat{\mathcal{I}}_{k}$ using the constrained matrix zonotope $\mathcal{N}_{\Theta_p}$ obtained in the offline phase, interval of all monomials $\Omega(\hat{\mathcal{I}}_{k},u(k))$, and the noise zonotope $\mathcal{Z}_w$. The computation in line \ref{ln:timeupdate} requires multiplying a constrained matrix zonotope $\mathcal{N}_{\Theta_p}$ by an interval $\Omega(\hat{\mathcal{I}}_{k},u(k))$ and Minkowski sum with a zonotope $\mathcal{Z}_w$. This can be over-approximated by either converting $\Omega(\hat{\mathcal{I}}_{k},u(k))$ to a matrix zonotope \cite{conf:ourjournal} and follow the traditional multiplication scheme in \cite{conf:cora,conf:thesisalthoff} or by converting the $\Omega(\hat{\mathcal{I}}_{k},u(k))$ to a zonotope and follow the multiplication proposed in \cite[Proposition 2]{conf:ourjournal}. In general, the results of the computations can be represented as a constrained zonotope $\tilde{\mathcal{R}}_{k}  = \zono{ \tilde{c}_k,\tilde{G}_k,\tilde{A}_k,\tilde{b}_k}$.    
  
\subsubsection{Measurement update}  

We use the implicit intersection approach presented in \cite{conf:data_driven_set_based}, where the measurement update set $\hat{\mathcal{R}}_{k}$ is determined directly from the time update set $\tilde{\mathcal{R}}_{k}$ and the measurements $y(k)$ in line \ref{ln:measupdate} as presented in the following proposition.

\begin{proposition}\!(\!\!\cite{conf:data_driven_set_based})
  The intersection of $\!{\tilde{\mathcal{R}}_{k} \! {=}\! \zono{\tilde{c}_k,\tilde{G}_k,\tilde{A}_k,\tilde{b}_k}\!}$ 
    and a region for $x(k)$ corresponding to $y(k)$
      as in \eqref{eq:observations} 
    can be described by the constrained zonotope 
    $\hat{\mathcal{R}}_k=\zono{\hat{c}_k,\hat{G}_k,\hat{A}_k,\hat{b}_k}$ 
 where $\hat{c}_{k} = \tilde{c}_{k}$, $\hat{G}_{k} = \tilde{G}_{k}$, and
  \begin{equation*}
      \hat{A}_k = 
      \begin{bmatrix}
        \tilde{A}_k  & 0  \\
      H \tilde{G}_k &G_{v}   \end{bmatrix}, 
      \hat{b}_k = 
      \begin{bmatrix}
        \tilde{b}_k \\
        y(k) - H {c}_k - c_{v}
      \end{bmatrix}. 
    \end{equation*}
  \vspace{-4mm}
\label{prop:con_zonotopes}
\end{proposition}

The reachable set computed in Algorithm \ref{alg:PolySetBasedEstimation} over-approximates the
exact reachable set, i.e., $\mathcal{R}_{k+1} \subseteq \hat{\mathcal{R}}_{k+1}$ due to the inclusion of the true $\Theta_p$ inside $\mathcal{M}_{\Theta_p}$ and accordingly inside  $\mathcal{N}_{\Theta_p}$ under the assumption that the side information \eqref{eq:sideinfo} holds for the true $\Theta_p$.

\begin{table*}[hbt!]
\caption{System model bounds.}
\label{tab:matzono}
\centering
\normalsize
\begin{tabular}{c  l l }
\toprule
  & Left bound & Right bound\\
\midrule
Matrix zonotope $\mathcal{M}_{\Theta_p}$ & 
$\left[{\scriptsize \begin{array}{rrrr} 
 0.23 & -0.17 & -0.23 & -0.70 \\ 
-1.29 &  0.63 & -0.35 & -1.20 \\
-1.61 & -0.41 &  0.58 & -1.30
\end{array}}\right]$
& 
$\left[{\scriptsize \begin{array}{cccc} 
1.79 & 0.23 & 0.19 & 0.66 \\ 
1.37 & 1.32 & 0.38 & 1.11 \\
1.50 & 0.39 & 1.43 & 1.41
\end{array}}\right]$  
\\ [1.5em]
Constrained matrix zonotope $\mathcal{N}_{\Theta_p}$ & $\left[{\scriptsize \begin{array}{crcr} 
1 &  0   &   0 & -0.69\\ 
0 & 0.63 &   0 & 0\\
0 & -0.30  & 1     &              0
\end{array}}\right]$ 
&  
$\left[{\scriptsize \begin{array}{cccc}
1 &    0 &  0 &  0\\ 
 0 &   1.30 &  0 &  1\\
0 &    0.39 &  1&  0
\end{array}}\right]$\\
\bottomrule
\end{tabular}
\end{table*}

\begin{figure*}[hbt!]
    \centering
    \begin{subfigure}[h]{0.3\textwidth}
      \centering
        \includegraphics[width=\textwidth]{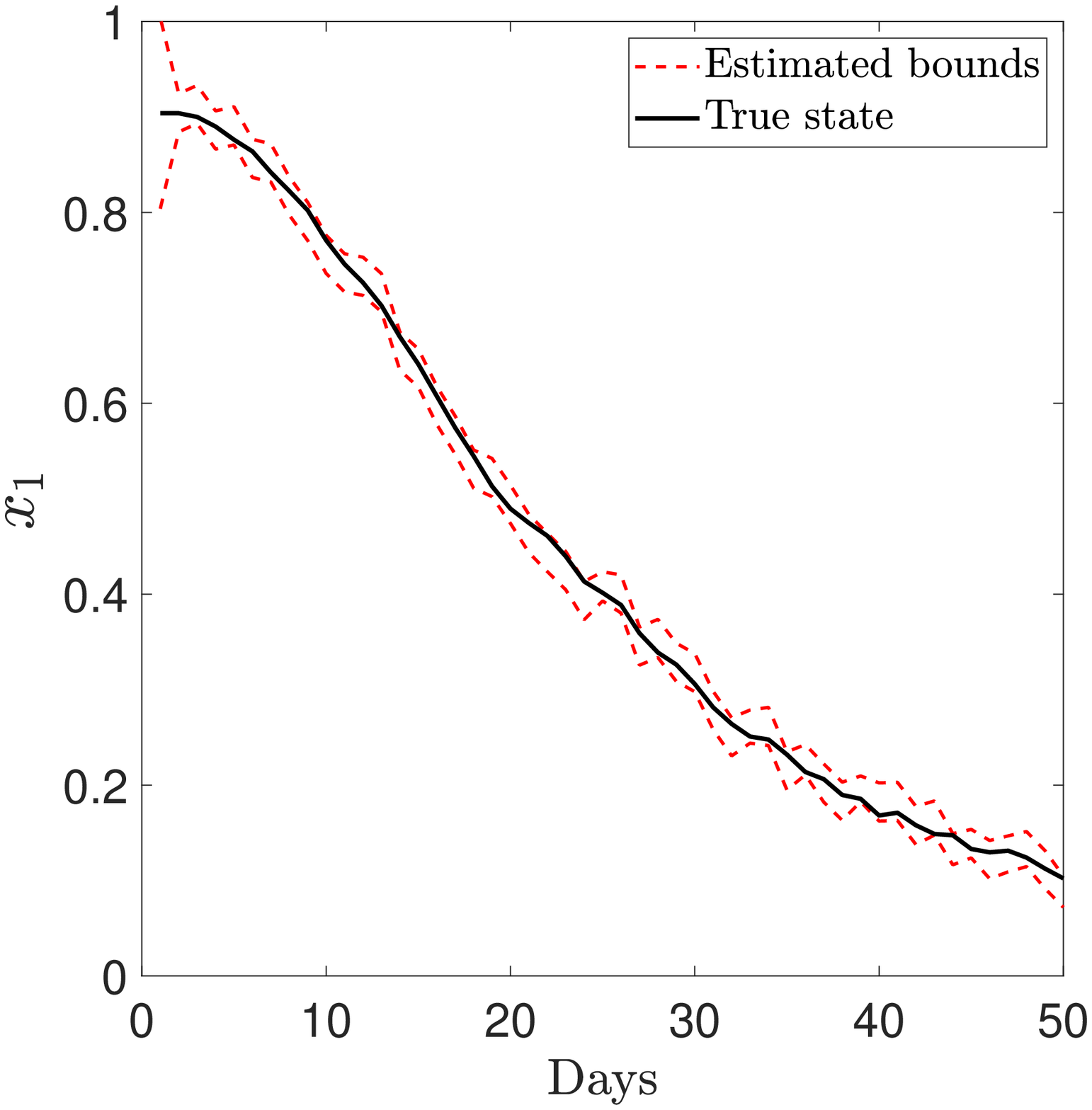}
        \caption{Susceptible $x_1(k)$}
        \label{fig:x1}
    \end{subfigure}
   \begin{subfigure}[h]{0.3\textwidth}
      \centering
        \includegraphics[width=\textwidth]{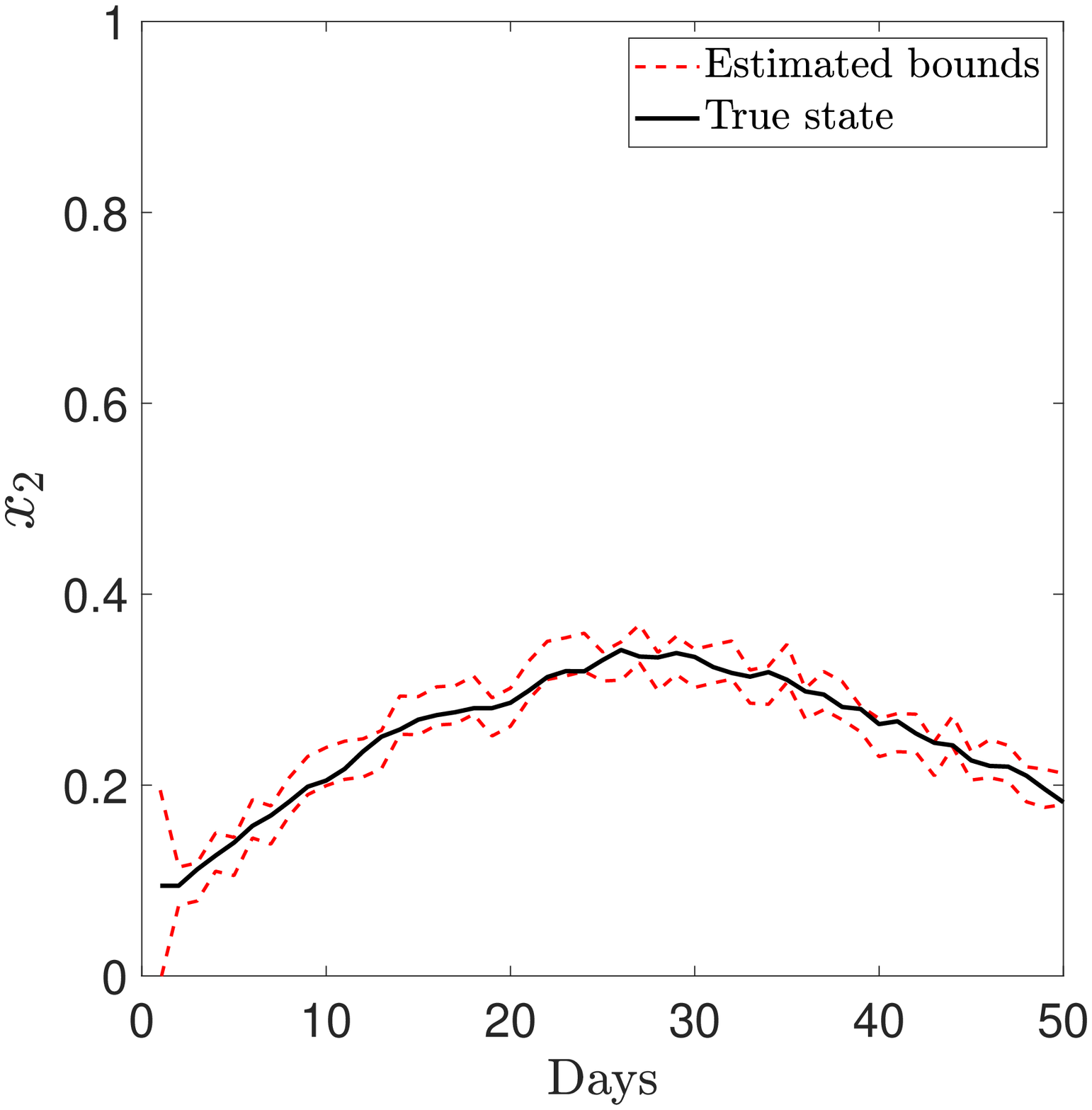}
        \caption{Infected $x_2(k)$}
        \label{fig:x2}
    \end{subfigure}
    \begin{subfigure}[h]{0.3\textwidth}
      \centering
        \includegraphics[width=\textwidth]{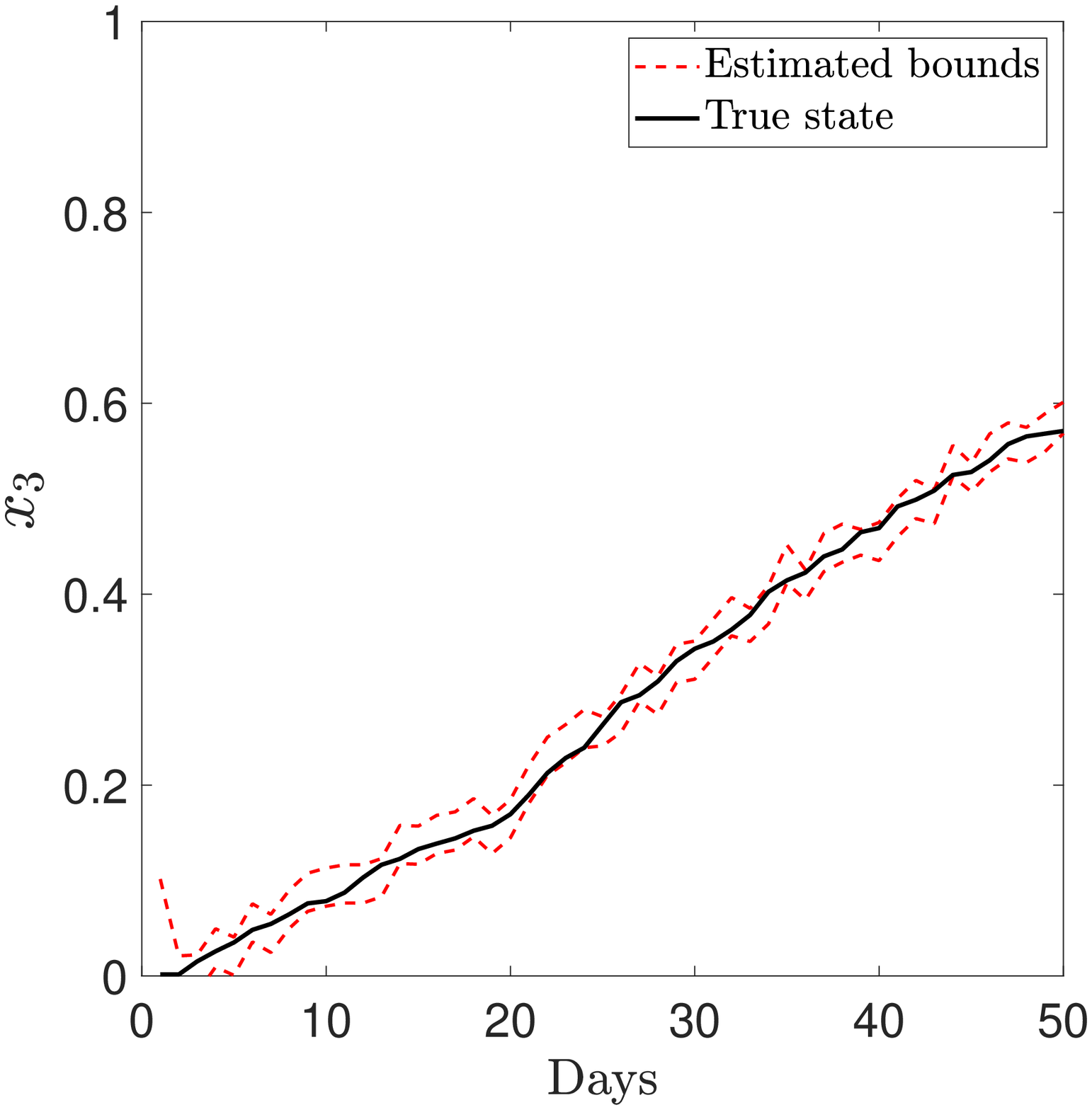}
        \caption{Recovered $x_3(k)$}
        \label{fig:x3}
    \end{subfigure}
\caption{Set-based state estimation of SIR epidemics.}
    \label{fig:threeresults}
    \vspace{-20pt}
\end{figure*}

\section{Application to SIR Epidemics} \label{sec:eval}

Consider a discrete-time SIR 
epidemic model with constant population
\[\begin{split}
    x_1(k+1) &= x_1(k) - \beta x_1(k) x_2(k) + w_1(k) \\
    x_2(k+1) &= x_2(k) + \beta x_1(k) x_2(k) - \gamma x_2(k) + w_2(k) \\
    x_3(k+1) &= x_3(k) + \gamma x_2(k) + w_3(k)
\end{split}\]
where $k=0,1,2,\dots$ are days; $x_1(k),x_2(k),x_3(k)\in[0,1]$ are respectively the proportions of susceptible, infected, and recovered populations; $\beta,\gamma\in [0,1]$ are the infection and recovery parameters; and $w_1(k),w_2(k),w_3(k)$ are the bounded process noise inputs. Notice that the discretization step $\bar{h}$ is assumed to be one day, therefore, if $\beta\in[0,1]$, then the condition $\bar{h}\beta\leq 1$ of \cite{mbio:allen1994} is satisfied. Considering
\begin{align}
\Theta_p &= \left[\begin{array}{cccr}
1 & 0 & 0 & -\beta \\
0 & 1-\gamma & 0 & \beta \\
0 & \gamma & 1 & 0
\end{array}\right] \label{eq:sirTheta_p} \\
h(x(k)) &= \left[\begin{array}{cccc}
    x_1(k) & x_2(k) & x_3(k) & x_1(k) x_2(k)
\end{array}\right]^\t \nonumber
\end{align}
then, from \eqref{eq:sysnonlin_poly} and \eqref{eq:pcg}, we can write the SIR model as
\[
x(k+1) = f_p(x(k)) + w(k) = \Theta_p h(x(k)) + w(k)
\]
where $x=[\colsep=2pt\begin{array}{ccc} x_1 & x_2 & x_3 \end{array}]^\t$ and $w=[\colsep=2pt\begin{array}{ccc} w_1 & w_2 & w_3 \end{array}]^\t$.
The output
\[
y(k) = \left[\begin{array}{c}
    x_2(k) \\ x_3(k) \\ x_1(k)+x_2(k)+x_3(k)
\end{array}\right] + \left[\begin{array}{c}
    v_1(k)  \\ v_2(k) \\ v_3(k) 
\end{array}\right]
\]
where $v_1(k),v_2(k),v_3(k)$ are the bounded measurement noise variables. The last output $x_1(k)+x_2(k)+x_3(k)\approx 1$ is due to the constant population assumption. For the SIR epidemic model, we need to measure $\gamma I$ for observability and identifiability \cite{hamelin2020}. However, since $\gamma$ is unknown, we need to measure at least two states to satisfy these properties.

We generate $Z$ as in \eqref{eq:data_sequences} using the true values of $\beta=0.3$ and $\gamma=0.1$, $x(0)=[\colsep=5pt\begin{array}{ccc} 0.9 & 0.09 & 0.01 \end{array}]^\t$, and $w(k) \in \zono{0, [\begin{array}{ccc} 0.001 & 0.001 & 0.001\end{array}]^\t}$ and $v(k) \in \zono{0, [\begin{array}{ccc} 0.001 & 0.001 & 0.001\end{array}]^\t}$. We use $Z$, without the knowledge of $\beta$ and $\gamma$, to find the constrained matrix zonotope $\mathcal{N}_{\Theta_p}$ and ensuring that $\Theta_p\in\mathcal{N}_{\Theta_p}$. First, we find $\mathcal{M}_{\Theta_p}$ using \eqref{eq:Msigma}, whose left and right bounds are given in Table~\ref{tab:matzono}. Second, as the structure of $\Theta_p$ is given in \eqref{eq:sirTheta_p}, we use this prior knowledge and impose it on the matrix zonotope $\mathcal{M}_{\Theta_p}$ by using the inequality \eqref{eq:sideinfo}, where 
\[
\bar{Y} = \left[\arraycolsep=3pt\begin{array}{cccr}
    1 & 0 & 0 & -0.5  \\
    0 & 1-0.3 & 0 & 0.5 \\
    0 & 0.3 & 0 & 1
\end{array}\right], \quad 
\bar{R} = \left[\arraycolsep=3pt\begin{array}{cccc}
    0 & 0 & 0 & 0.5  \\
    0 & 0.5 & 0 & 0.5 \\
    0 & 0.5 & 0 & 0
\end{array}\right]
\]
are the initial guess on $\Theta_p$ and the maximum uncertainty on $\bar{Y}$, and $\bar{Q}=I_{n_x}$. The matrix $\bar{R}$ is chosen as such by keeping into account the prior knowledge that $\beta,\gamma\in[0,1]$. From the lines \ref{ln:A1}--\ref{ln:N} of Algorithm~\ref{alg:PolySetBasedEstimation}, we finally obtain $\mathcal{N}_{\Theta_p}$, whose left and right bounds are given in Table~\ref{tab:matzono}. From $\mathcal{N}_{\Theta_p}$, we obtain $-0.69\leq-\beta\leq 0$, which gives $\beta\in[0,0.69]$, and ${0.63\leq 1-\gamma \leq 1.30}$, intersecting it with prior knowledge $\gamma\in[0,1]$ gives $\gamma\in[0,0.37]$. This validates Proposition~\ref{prop:sigmaM_p} as $\Theta_p$ with true $\beta=0.3$ and $\gamma=0.1$ is contained in $\mathcal{N}_{\Theta_p}$.

For the online phase, we first suppose that the region $\mathcal{X}_0$ for the initial condition is known, where it is guessed that $x_1(0)\in[0.8,1]$, $x_2(0)\in[0,0.2]$, and $x_3(0)\in[0,0.1]$. We first consider $\hat{\mathcal{R}}_0=\mathcal{X}_0$, then, at each time step $k$, using the noisy online measurements $y(k)$, we update the set estimate $\hat{\mathcal{R}}_{k}$ using the lines~\ref{ln:algI_i}--\ref{ln:measupdate} of Algorithm~\ref{alg:PolySetBasedEstimation}. These set-based estimation results are depicted in Figure~\ref{fig:threeresults}, where it is ensured that the true state $x(k)$ remains bounded by the constrained zonotope $\hat{\mathcal{R}}_k$. This validates Proposition~\ref{prop:con_zonotopes}.

\section{Conclusion} \label{sec:conc}

We presented a data-driven set-based estimator for polynomial systems. In the offline phase, a set of models is computed that is consistent with the experimental data and the noise bounds. In the online phase, the time update step involves propagating ahead of the estimated set using the set of consistent models. Then, the measurement update step computes the set-based estimate by intersecting the time update set with the set consistent with noisy measurements. We evaluate the proposed approach on SIR epidemics by estimating sets that bound the model parameters and the proportion of susceptible, infected, and recovered populations under bounded uncertainties and measurement noise.

Future investigations include the development of methods that further reduce the conservativeness of estimated sets without violating the formal guarantees. Moreover, observability conditions of polynomial systems need to be examined when the model is partially known.


\section*{Acknowledgement}

This work was supported by the Swedish Research Council, the Knut and Alice Wallenberg Foundation, the Democritus project on Decision-making in Critical Societal Infrastructures by Digital Futures, and the European Unions Horizon 2020 Research and Innovation program under the CONCORDIA cyber security project (GA No. 830927).


\bibliographystyle{ieeetr}
\bibliography{ref} 

\begin{thebibliography}{10}

\bibitem{conf:deepctherom}
J.~Coulson, J.~Lygeros, and F.~Dorfler, ``Distributionally robust chance
  constrained data-enabled predictive control,'' {\em IEEE Transactions on
  Automatic Control}, 2021.
\newblock Available Online.

\bibitem{huang2021robust}
L.~Huang, J.~Zhen, J.~Lygeros, and F.~D{\"o}rfler, ``Robust data-enabled
  predictive control: Tractable formulations and performance guarantees,'' {\em
  arXiv preprint arXiv:2105.07199}, 2021.

\bibitem{conf:setmem1971}
D.~Bertsekas and I.~Rhodes, ``Recursive state estimation for a set-membership
  description of uncertainty,'' {\em IEEE Transactions on Automatic Control},
  vol.~16, no.~2, pp.~117--128, 1971.

\bibitem{efimov2013}
D.~Efimov, T.~Ra{\"\i}ssi, S.~Chebotarev, and A.~Zolghadri, ``Interval state
  observer for nonlinear time varying systems,'' {\em Automatica}, vol.~49,
  no.~1, pp.~200--205, 2013.

\bibitem{conf:combastelinterval2}
R.~E.~H. Thabet, T.~Raissi, C.~Combastel, D.~Efimov, and A.~Zolghadri, ``An
  effective method to interval observer design for time-varying systems,'' {\em
  Automatica}, vol.~50, no.~10, pp.~2677--2684, 2014.

\bibitem{conf:combastelinterval1}
H.~Ethabet, T.~Ra{\"\i}ssi, M.~Amairi, C.~Combastel, and M.~Aoun, ``Interval
  observer design for continuous-time switched systems under known switching
  and unknown inputs,'' {\em International Journal of Control}, vol.~93, no.~5,
  pp.~1088--1101, 2020.

\bibitem{conf:orthotope}
G.~Belforte, B.~Bona, and V.~Cerone, ``Parameter estimation algorithms for a
  set-membership description of uncertainty,'' {\em Automatica}, vol.~26,
  no.~5, pp.~887--898, 1990.

\bibitem{conf:ierardi_setmem}
C.~Ierardi, L.~Orihuela, and I.~Jurado, ``A distributed set-membership
  estimator for linear systems with reduced computational requirements,'' {\em
  Automatica}, vol.~132, p.~109802, 2021.

\bibitem{conf:ourdiffusion}
A.~Alanwar, J.~J. Rath, H.~Said, K.~H. Johansson, and M.~Althoff, ``Distributed
  set-based observers using diffusion strategy,'' {\em arXiv preprint
  arXiv:2003.10347}, 2020.

\bibitem{alamo2005}
T.~Alamo, J.~M. Bravo, and E.~F. Camacho, ``Guaranteed state estimation by
  zonotopes,'' {\em Automatica}, vol.~41, no.~6, pp.~1035--1043, 2005.

\bibitem{le2013}
V.~T.~H. Le, C.~Stoica, T.~Alamo, E.~F. Camacho, and D.~Dumur, ``Zonotopic
  guaranteed state estimation for uncertain systems,'' {\em Automatica},
  vol.~49, no.~11, pp.~3418--3424, 2013.

\bibitem{rego2020}
B.~S. Rego, G.~V. Raffo, J.~K. Scott, and D.~M. Raimondo, ``Guaranteed methods
  based on constrained zonotopes for set-valued state estimation of nonlinear
  discrete-time systems,'' {\em Automatica}, vol.~111, p.~108614, 2020.

\bibitem{rego2021}
B.~S. Rego, J.~K. Scott, D.~M. Raimondo, and G.~V. Raffo, ``Set-valued state
  estimation of nonlinear discrete-time systems with nonlinear invariants based
  on constrained zonotopes,'' {\em Automatica}, vol.~129, p.~109638, 2021.

\bibitem{conf:AlthoffJagatjournal}
M.~Althoff and J.~J. Rath, ``Comparison of guaranteed state estimators for
  linear time-invariant systems,'' {\em Automatica}, vol.~130, 2021.
\newblock article no. 109662.

\bibitem{paula2022}
A.~A. de~Paula, G.~V. Raffo, and B.~O. Teixeira, ``Zonotopic filtering for
  uncertain nonlinear systems: Fundamentals, implementation aspects, and
  extensions [applications of control],'' {\em IEEE Control Systems Magazine},
  vol.~42, no.~1, pp.~19--51, 2022.

\bibitem{conf:LearningDensity}
Y.~Meng, D.~Sun, Z.~Qiu, M.~T.~B. Waez, and C.~Fan, ``Learning density
  distribution of reachable states for autonomous systems,'' in {\em 5th Annual
  Conference on Robot Learning}, 2021.

\bibitem{conf:onthefly}
F.~Djeumou, A.~P. Vinod, E.~Goubault, S.~Putot, and U.~Topcu, ``On-the-fly
  control of unknown smooth systems from limited data,'' in {\em American
  Control Conference}, pp.~3656--3663, 2021.

\bibitem{conf:murat_christoffel}
A.~Devonport, F.~Yang, L.~E. Ghaoui, and M.~Arcak, ``{Data-Driven Reachability
  Analysis with Christoffel Functions},'' {\em arXiv preprint
  arXiv:2104.13902}, 2021.

\bibitem{conf:activelearning1}
A.~Chakrabarty, A.~Raghunathan, S.~{Di Cairano}, and C.~Danielson,
  ``{Data-driven estimation of backward reachable and invariant sets for
  unmodeled systems via active learning},'' in {\em IEEE Conference on Decision
  and Control}, pp.~372--377, 2018.

\bibitem{conf:uncertain}
A.~R.~R. Matavalam, U.~Vaidya, and V.~Ajjarapu, ``Data-driven approach for
  uncertainty propagation and reachability analysis in dynamical systems,'' in
  {\em American Control Conference}, pp.~3393--3398, 2020.

\bibitem{conf:stochasticreach}
A.~J. Thorpe, K.~R. Ortiz, and M.~M. Oishi, ``Data-driven stochastic
  reachability using hilbert space embeddings,'' {\em arXiv preprint
  arXiv:2010.08036}, 2020.

\bibitem{conf:koopmanblack}
S.~Bak, S.~Bogomolov, P.~S. Duggirala, A.~R. Gerlach, and K.~Potomkin,
  ``Reachability of black-box nonlinear systems after koopman operator
  linearization,'' {\em arXiv preprint arXiv:2105.00886}, 2021.

\bibitem{conf:l4dc}
A.~Alanwar, A.~Koch, F.~Allg{\"o}wer, and K.~H. Johansson, ``Data-driven
  reachability analysis using matrix zonotopes,'' in {\em Learning for Dynamics
  and Control}, pp.~163--175, 2021.

\bibitem{conf:zpc}
A.~Alanwar, Y.~St{\"u}rz, and K.~H. Johansson, ``Robust data-driven predictive
  control using reachability analysis,'' {\em arXiv preprint arXiv:2103.14110},
  2021.

\bibitem{conf:data_driven_set_based}
A.~Berndt, A.~Alanwar, K.~H. Johansson, and H.~Sandberg, ``Data-driven
  set-based estimation using matrix zonotopes with set containment
  guarantees,'' {\em arXiv preprint arXiv:2101.10784}, 2021.

\bibitem{bliman2016}
P.-A. Bliman and B.~D. Barros, ``Interval observers for {SIR} epidemic models
  subject to uncertain seasonality,'' in {\em International Symposium on
  Positive Systems}, pp.~31--39, 2016.

\bibitem{degue2020}
K.~H. Degue and J.~Le~Ny, ``Estimation and outbreak detection with interval
  observers for uncertain discrete-time {SEIR} epidemic models,'' {\em
  International Journal of Control}, vol.~93, no.~11, pp.~2707--2718, 2020.

\bibitem{efimov2021}
D.~Efimov and R.~Ushirobira, ``On an interval prediction of {COVID-19}
  development based on a {SEIR} epidemic model,'' {\em Annual Reviews in
  Control}, vol.~51, pp.~477--487, 2021.

\bibitem{conf:zono1998}
W.~K{\"u}hn, ``Rigorously computed orbits of dynamical systems without the
  wrapping effect,'' {\em Computing}, vol.~61, no.~1, pp.~47--67, 1998.

\bibitem{conf:const_zono}
J.~K. Scott, D.~M. Raimondo, G.~R. Marseglia, and R.~D. Braatz, ``Constrained
  zonotopes: A new tool for set-based estimation and fault detection,'' in {\em
  Automatica}, vol.~69, pp.~126--136, 2016.

\bibitem{conf:thesisalthoff}
M.~Althoff, {\em Reachability analysis and its application to the safety
  assessment of autonomous cars}.
\newblock PhD thesis, Technische Universit{\"a}t M{\"u}nchen, 2010.

\bibitem{conf:ourjournal}
A.~Alanwar, A.~Koch, F.~Allg{\"o}wer, and K.~H. Johansson, ``Data-driven
  reachability analysis from noisy data,'' {\em arXiv preprint
  arXiv:2105.07229}, 2021.

\bibitem{conf:inverseInterval}
M.~Fiedler, J.~Nedoma, J.~Ram{\'\i}k, J.~Rohn, and K.~Zimmermann, {\em Linear
  optimization problems with inexact data}.
\newblock Springer Science \& Business Media, 2006.

\bibitem{conf:inverseIntervalSurvey}
J.~Rohn and R.~Farhadsefat, ``Inverse interval matrix: a survey,'' {\em The
  Electronic Journal of Linear Algebra}, vol.~22, pp.~704--719, 2011.

\bibitem{albertini2002}
F.~Albertini and D.~D'Alessandro, ``Observability and forward-backward
  observability of discrete-time nonlinear systems,'' {\em Mathematics of
  Control, Signals and Systems}, vol.~15, no.~4, pp.~275--290, 2002.

\bibitem{hanba2009}
S.~Hanba, ``On the ``uniform'' observability of discrete-time nonlinear
  systems,'' {\em IEEE Transactions on Automatic Control}, vol.~54, no.~8,
  pp.~1925--1928, 2009.

\bibitem{conf:polyDissipat}
T.~Martin and F.~Allg{\"o}wer, ``Dissipativity verification with guarantees for
  polynomial systems from noisy input-state data,'' {\em IEEE Control Systems
  Letters}, vol.~5, no.~4, pp.~1399--1404, 2020.

\bibitem{conf:cora}
M.~Althoff, ``An introduction to {CORA} 2015,'' in {\em Proceedings of the
  Workshop on Applied Verification for Continuous and Hybrid Systems}, 2015.

\bibitem{mbio:allen1994}
L.~J.~S. Allen, ``Some discrete-time {SI}, {SIR}, and {SIS} epidemic models,''
  {\em Mathematical biosciences}, vol.~124, no.~1, pp.~83--105, 1994.

\bibitem{hamelin2020}
F.~Hamelin, A.~Iggidr, A.~Rapaport, and G.~Sallet, ``Observability,
  identifiability and epidemiology -- a survey,'' {\em arXiv preprint
  arXiv:2011.12202}, 2020.

\end{thebibliography}

\end{document}